\documentclass[submission,copyright,creativecommons,sharealike,noncommercial]{eptcs}
\providecommand{\event}{- - - -}
\pdfoutput=1

\usepackage{mathtools} 
\usepackage{amstext} 
\usepackage{amssymb} 
\usepackage{bbold} 
\usepackage{amsthm} 
\usepackage{stmaryrd} 

\usepackage{relsize} 
\usepackage{microtype} 
\usepackage{csquotes} 
\usepackage{cancel}

\usepackage{hyperref} 
\usepackage[nocompress]{cite} 

\usepackage{graphicx} 
\usepackage[usenames,dvipsnames]{xcolor} 
\usepackage{stackengine}
\usepackage{tikz} 
\usepackage{circuitikz} 
\usetikzlibrary{
	arrows,
	shapes,
	decorations,
	intersections,
	backgrounds,
	positioning,
	circuits.ee.IEC
	}

		\newcounter{theorem_c} 

		\theoremstyle{plain}

		\newtheorem{proposition}[theorem_c]{Proposition}
		
		\newtheoremstyle{exampstyle}
		  {2mm} 
		  {2mm} 
		  {\itshape} 
		  {} 
		  {\bfseries} 
		  {.} 
		  {.5em} 
		  {} 

		\theoremstyle{exampstyle}
		\newtheorem{definition}[theorem_c]{Definition}

	\newcommand{\inlineQuote}[1]{\textquotedblleft #1\textquotedblright} 

	\newcommand{\reals}{\mathbb{R}} 
	\newcommand{\complexs}{\mathbb{C}} 
	\newcommand{\integersMod}[1]{\mathbb{Z}_{#1}} 
	
		\newcommand{\ket}[1]{\vert #1 \rangle} 
		\newcommand{\bra}[1]{\langle #1 \vert} 
		\newcommand{\braket}[2]{\langle #1 \vert #2 \rangle} 
		
		\newcommand{\decohSym}{\operatorname{dec}} 
		\newcommand{\decoh}[1]{\decohSym_{#1}} 

		\newcommand{\isom}{\cong} 
		\newcommand{\id}[1]{id_{#1}} 

		\newcommand{\RMatCategory}[1]{#1\operatorname{-Mat}} 
		\newcommand{\fHilbCategory}{\operatorname{fHilb}} 
		\newcommand{\fdHilbCategory}{\fHilbCategory} 

		\newcommand{\CPMCategory}[1]{\operatorname{CPM}[#1]} 

		\newcommand{\KaroubiEnvelope}[1]{\operatorname{Split}\left[#1\right]} 
	\newcommand{\Xcolour}{Red}
	
	\newcommand{\Zcolour}{YellowGreen}

	\newcommand{\Xaltcolour}{Purple}
	
	\newcommand{\Zaltcolour}{Cyan}
	
	\newcommand{\Dcolour}{black!80}

	\newcommand{\Xbwcolour}{black!80}

	\newcommand{\Zbwcolour}{white}

	\newcommand{\Ybwcolour}{black!15}
	\newcommand{\YbwdotSym}{\hbox{\input{symbols/YbwdotSym.tex}}\!\!} 

	\newcommand{\Wbwcolour}{black!50}

	\newcommand{\trace}[1]{\hbox{\input{symbols/traceSym.tex}}\!_{#1}} 

	\tikzset{
	  rectangle with rounded corners north west/.initial=4pt,
	  rectangle with rounded corners south west/.initial=4pt,
	  rectangle with rounded corners north east/.initial=4pt,
	  rectangle with rounded corners south east/.initial=4pt,
	}
	\makeatletter
	\pgfdeclareshape{rectangle with rounded corners}{
	  \inheritsavedanchors[from=rectangle] 
	  \inheritanchorborder[from=rectangle]
	  \inheritanchor[from=rectangle]{center}
	  \inheritanchor[from=rectangle]{north}
	  \inheritanchor[from=rectangle]{south}
	  \inheritanchor[from=rectangle]{west}
	  \inheritanchor[from=rectangle]{east}
	  \inheritanchor[from=rectangle]{north east}
	  \inheritanchor[from=rectangle]{south east}
	  \inheritanchor[from=rectangle]{north west}
	  \inheritanchor[from=rectangle]{south west}
	  \backgroundpath{
	    \southwest \pgf@xa=\pgf@x \pgf@ya=\pgf@y
	    \northeast \pgf@xb=\pgf@x \pgf@yb=\pgf@y
	    \pgfkeysgetvalue{/tikz/rectangle with rounded corners north west}{\pgf@rectc}
	    \pgfsetcornersarced{\pgfpoint{\pgf@rectc}{\pgf@rectc}}
	    \pgfpathmoveto{\pgfpoint{\pgf@xa}{\pgf@ya}}
	    \pgfpathlineto{\pgfpoint{\pgf@xa}{\pgf@yb}}
	    \pgfkeysgetvalue{/tikz/rectangle with rounded corners north east}{\pgf@rectc}
	    \pgfsetcornersarced{\pgfpoint{\pgf@rectc}{\pgf@rectc}}
	    \pgfpathlineto{\pgfpoint{\pgf@xb}{\pgf@yb}}
	    \pgfkeysgetvalue{/tikz/rectangle with rounded corners south east}{\pgf@rectc}
	    \pgfsetcornersarced{\pgfpoint{\pgf@rectc}{\pgf@rectc}}
	    \pgfpathlineto{\pgfpoint{\pgf@xb}{\pgf@ya}}
	    \pgfkeysgetvalue{/tikz/rectangle with rounded corners south west}{\pgf@rectc}
	    \pgfsetcornersarced{\pgfpoint{\pgf@rectc}{\pgf@rectc}}
	    \pgfpathclose
	 }
	}
	\makeatother

	\tikzset{->-/.style={decoration={markings,mark=at position #1 with {\arrow{>}}},postaction={decorate}}}
	\tikzset{-<-/.style={decoration={markings,mark=at position #1 with {\arrow{<}}},postaction={decorate}}}

	\tikzstyle{every picture}=[baseline=-0.25em,scale=0.5]
	\pgfdeclarelayer{edgelayer}
	\pgfdeclarelayer{nodelayer}
	\pgfsetlayers{edgelayer,nodelayer,main}

	\tikzstyle{box} = [draw,shape=rectangle,inner sep=2pt,minimum height=6mm,minimum width=6mm,fill=white] 
	\tikzstyle{boxlarge} = [draw,shape=rectangle,inner sep=2pt,minimum height=1.5cm,minimum width=8mm,fill=white] 
	\tikzstyle{boxLarge} = [draw,shape=rectangle,inner sep=2pt,minimum height=2cm,minimum width=10mm,fill=white] 
	\tikzstyle{boxsmall} = [draw,shape=rectangle,inner sep=2pt,minimum height=3mm,minimum width=3mm,fill=white] 
	\tikzstyle{dot} = [inner sep=0mm,minimum width=3mm,minimum height=3mm,draw,shape=circle,text depth=-0.1mm]
	\tikzstyle{Zbwdot} = [dot, fill=\Zbwcolour]
	\tikzstyle{Xbwdot} = [dot, fill=\Xbwcolour]
	\tikzstyle{Ybwdot} = [dot, fill=\Ybwcolour]
	\tikzstyle{Wbwdot} = [dot, fill=\Wbwcolour]
	\tikzstyle{antipode} = [boxsmall] 

	\tikzstyle{state} = [draw, rectangle with rounded corners,
	  rectangle with rounded corners north west=8pt,
	  rectangle with rounded corners south west=8pt,
	  rectangle with rounded corners north east=0pt,
	  rectangle with rounded corners south east=0pt,
	,inner sep=2pt,minimum height=6mm,minimum width=6mm,fill=white]
	\tikzstyle{statelarge} = [draw, rectangle with rounded corners,
	  rectangle with rounded corners north west=8pt,
	  rectangle with rounded corners south west=8pt,
	  rectangle with rounded corners north east=0pt,
	  rectangle with rounded corners south east=0pt,
	,inner sep=2pt,minimum height=1.5cm,minimum width=8mm,fill=white]
	\tikzstyle{stateLarge} = [draw, rectangle with rounded corners,
	  rectangle with rounded corners north west=8pt,
	  rectangle with rounded corners south west=8pt,
	  rectangle with rounded corners north east=0pt,
	  rectangle with rounded corners south east=0pt,
	,inner sep=2pt,minimum height=2cm,minimum width=8mm,fill=white]
	\tikzstyle{effect} = [draw, rectangle with rounded corners,
	  rectangle with rounded corners north west=0pt,
	  rectangle with rounded corners south west=0pt,
	  rectangle with rounded corners north east=8pt,
	  rectangle with rounded corners south east=8pt,
	,inner sep=2pt,minimum height=6mm,minimum width=6mm,fill=white]
	\tikzstyle{scalar}=[diamond,draw,inner sep=1pt,font=\small,fill=white]

	\tikzstyle{cdnode}=[fill=white]
	\tikzstyle{labelnode}=[fill=white]
	\tikzstyle{tightlabelnode}=[fill=white,inner sep = 0.1mm]
	\tikzstyle{none}=[inner sep=0pt]
	\tikzstyle{whiteline}=[-, line width=4pt, draw=white]

	\tikzstyle{trace}=[circuit ee IEC,thick,ground,scale=2.5]
	\tikzstyle{cotrace}=[circuit ee IEC,thick,ground,rotate=180,scale=2.5]
	\tikzstyle{upground}=[circuit ee IEC,thick,ground,rotate=90,scale=2.5]
	\tikzstyle{downground}=[circuit ee IEC,thick,ground,rotate=-90,scale=2.5]

	\tikzstyle{doubled} = [line width=1.8pt] 
	
	\tikzstyle{empty diagram}=[draw=gray!40!white,dashed,shape=rectangle,minimum width=1cm,minimum height=1cm]

\definecolor{red}{RGB}{240,0,0}
\definecolor{green}{RGB}{36,255,36}
\definecolor{blue}{RGB}{50,122,195}
\definecolor{RoyalPurple}{RGB}{73,0,146}

\title{Density Hypercubes, Higher Order Interference \\ and Hyper-Decoherence: a Categorical Approach}
\author{
	Stefano Gogioso\\
	University of Oxford \\
	\texttt{stefano.gogioso@cs.ox.ac.uk}
	\and
	Carlo Maria Scandolo \\
	University of Oxford \\
	\texttt{carlomaria.scandolo@cs.ox.ac.uk}
}

\def\titlerunning{Density Hypercubes, Higher Order Interference and Hyper-Decoherence: a Categorical Approach}
\def\authorrunning{S. Gogioso \& C. M. Scandolo}

\newcommand{\defi}[1]{\emph{#1}}

\newcommand{\Aut}[1]{\textnormal{Aut(}#1\textnormal{)}}

\newcommand{\DDCategorySym}{\textnormal{DD}}
\newcommand{\DDCategory}[1]{\DDCategorySym\textnormal{(}#1\textnormal{)}}
\newcommand{\DHCategory}[1]{\textnormal{DH(}#1\textnormal{)}}

\newcommand{\hypdecohSym}{\textnormal{hypdec}} 
\newcommand{\hypdecoh}[1]{\hypdecohSym_{#1}} 

\renewcommand{\fHilbCategory}{{\textnormal{fHilb}}}
\renewcommand{\RMatCategory}[1]{{#1\textnormal{-Mat}}}
\renewcommand{\CPMCategory}[1]{{\textnormal{CPM(}#1\textnormal{)}}}
\renewcommand{\KaroubiEnvelope}[1]{\textnormal{Split(}#1\textnormal{)}}

\renewcommand{\decohSym}{\textnormal{dec}} 
\renewcommand{\decoh}[1]{\decohSym_{#1}} 

\renewcommand{\Zbwcolour}{white}
\renewcommand{\Xbwcolour}{black!80}

\begin{document}

\maketitle

\begin{abstract}
	In this work, we use the recently introduced double-dilation construction by Zwart and Coecke to construct a new categorical probabilistic theory of density hypercubes. By considering multi-slit experiments, we show that the theory displays higher-order interference of order up to fourth. We also show that the theory possesses hyperdecoherence maps, which can be used to recover quantum theory in the Karoubi envelope. 
\end{abstract}

\section{Introduction} 
\label{section_introduction}

Quantum interference is often considered to be one of the fundamental features of quantum theory, responsible for quantum advantage in a number of computational tasks. However, there is a known limit to how much interference quantum theory can exhibit. Sorkin proposed a hierarchy of theories based on the maximum order of interference they exhibit \cite{Sorkin1,Sorkin2}, which is quantified by the maximum number of slits on which a theory shows an irreducible interference behaviour. Interference in quantum theory is limited to the second order: the interference pattern of two slits cannot be reduced to the pattern of single slits, but the interference pattern of three slits can be reduced to the pattern arising from pairs of slits and single slits. This limitation has been recently confirmed in various experiments \cite{sinha2010ruling,park2012three,sinha2015superposition,kauten2015obtaining,jin2017experimental}.

A natural question arises: Why is interference in Nature limited to the second order? Does the presence of higher-order interference create any paradoxical consequences in Nature that conflict with some of the principles we believe to be fundamental? Recent work has shown that higher-order interference---i.e. interference of order higher than the second---is forbidden \cite{HOP} in physical theories which admit a fundamental level of description where everything is pure and reversible \cite{TowardsThermo,Purity}. Further work has ruled out higher-order interference based on thermodynamic considerations \cite{Barnum-thermo,TowardsThermo}.

Other literature has instead focused on the analysis of specific feature that theories with higher-order interference would possess, e.g.\ whether they would provide any advantage in certain computational tasks \cite{Lee-Selby-interference,Control-reversible,Lee-Selby-Grover,Oracles}. It was also shown that theories having second-order interference and lacking interference of higher orders are relatively close to quantum theory \cite{Barnum-interference,Ududec-3slits,CozThesis,Niestegge}.

Unfortunately, one of the major shortcomings in the study of higher-order interference is the scarcity of concrete models displaying such post-quantum features, so that it has so far been very hard to look for specific examples of paradoxical or counter-intuitive consequences. Two models---density cubes \cite{Density-cubes} and quartic quantum theory \cite{Quartic-theory}---have been proposed in the past, but are not fully defined operational theories, e.g.\ because they do not deal with composite systems \cite{Lee-Selby-interference}. This limitation precludes them from being used to study all possible consequences of higher order interference, including potential violation of Tsirelson bound.

In this article, we provide the first complete construction of a full-fledged operational theory exhibiting interference up to the fourth order. Our construction is inspired by the double-dilation construction of \cite{double-mixing} 
and it is carried out in within the framework of categorical probabilistic theories \cite{gogioso2017categorical}. The resulting theory of `density hypercubes' has composite systems, exhibits higher-order interference and possesses hyper-decoherence maps \cite{Quartic-theory,Lee-Selby-interference,lee2017no}. Quantum theory, with its second-order interference, is an extension of classical theory: the latter can be recovered by decoherence, which eliminates the second-order interference effects. Similarly, the theory of density hypercubes, with its third- and fourth-order interference, is an extension of quantum theory: the latter can now be recovered by hyper-decoherence, which eliminates third- and fourth-order interference effects. 

The paper is organized as follows. 
In Section \ref{section_densityHypercubes}, we define the categorical probabilistic theory of density hypercubes using the double-dilation construction. 
In Section \ref{section_hyperDecoherence}, we define hyper-decoherence maps, and show that quantum theory is recovered in the Karoubi envelope. 
In Section \ref{section_higherOrderInterference}, we show that density hypercubes display interference of third- and fourth-order, but not of fifth-order and above. 
In Section \ref{section_conclusions}, finally, we discuss open questions and future lines of research. 
Proofs of all results can be found in the Appendix.
\vspace{12pt} 

\section{The Theory of Density Hypercubes}
\label{section_densityHypercubes}

\subsection{Construction of the theory}

In this section, we define the categorical probabilistic theory of \defi{density hypercubes}, using a recently introduced construction known as \defi{double dilation} \cite{double-mixing}. The construction is done in two steps: first we define the category \DDCategory{\fHilbCategory}, containing hyper-quantum systems and processes between them, and only in a second moment we introduce quantum and classical systems, using (hyper-)decoherence and working in the Karoubi envelope \KaroubiEnvelope{\DDCategory{\fHilbCategory}}.

The \defi{double-dilation category} \DDCategory{\fHilbCategory} is defined to be a symmetric monoidal subcategory of \CPMCategory{\fHilbCategory} with objects---the \defi{density hypercubes}---in the form $\DDCategory{H}:= \mathcal{H} \otimes \mathcal{H}$, where $H$ is a finite-dimensional Hilbert space and $\mathcal{H}:=H^\ast \otimes H$ is the corresponding doubled system in the CPM category. Even though \DDCategory{\fHilbCategory} is symmetric monoidal and has its own graphical calculus, in this work we will always use the graphical calculi of \CPMCategory{\fHilbCategory} and \fHilbCategory{} to talk about density hypercubes. When working in \CPMCategory{\fHilbCategory}, we will use solid black lines for morphisms and calligraphic letters (e.g. $\mathcal{H}$) for objects. When working in \fHilbCategory, we will use solid grey lines for morphisms and plain letters (e.g. $H$) for objects. 

The morphisms $\DDCategory{H} \rightarrow \DDCategory{K}$ in \DDCategory{\fHilbCategory} are the CP maps $\mathcal{H}\otimes \mathcal{H} \rightarrow \mathcal{K} \otimes\mathcal{K}$ taking the following form for a doubled CP map $F$, some auxiliary systems $\mathcal{E},\mathcal{G}$ and some special commutative $\dagger$-Frobenius algebra $\hbox{\input{symbols/ZdotSym.tex}}\!$ (henceforth known as a \defi{classical structure}) on $G$ in \fHilbCategory:
\begin{equation}
\scalebox{0.8}{$
	\input{pictures/doublyMixedCPmap.tikz}
$}
\end{equation}
In the diagram above, $F$ is a doubled CP map $\mathcal{H} \rightarrow \mathcal{G} \otimes \mathcal{K} \otimes \mathcal{E}$ in \CPMCategory{\fHilbCategory}---i.e. one in the form $F = f^\ast \otimes f$ for some $f:H \rightarrow G \otimes K \otimes E$ in \fHilbCategory---and we have used $\bar{F}$ to denote the CP map obtained by inverting the tensor product ordering of inputs and outputs of $f$ (for purely aesthetic reasons). We will always use upper-case letters (e.g. $F$) to denote doubled CP maps in \CPMCategory{\fHilbCategory}, lower-case letters to denote the corresponding linear maps in \fHilbCategory, and we will always write discarding maps explicitly.

Composition in \DDCategory{\fHilbCategory} is the same as composition of CP maps, while tensor product is only slightly adjusted to take into account the doubled format of our new morphisms:
\begin{equation}
\scalebox{0.5}{$
	\input{pictures/tensorProductDHmaps1.tikz}
	\hspace{3mm} \bigotimes \hspace{3mm}
	\input{pictures/tensorProductDHmaps2.tikz}
	\hspace{3mm} = \hspace{3mm}
	\input{pictures/tensorProductDHmaps3.tikz}
	$}
\end{equation}
Just as was the case for CP maps, maps of density hypercubes can all be obtained as composition of a ``doubled'' map and one or two ``discarding'' maps:
\begin{equation}
	\underbrace{\addstackgap[6pt]{$
		\scalebox{0.8}{$\input{pictures/DHpureDiscardingMaps1.tikz}$}
	$}}_{\text{doubled map}}
	\hspace{3cm}
	\underbrace{\addstackgap[6pt]{$
		\scalebox{0.8}{$\input{pictures/DHpureDiscardingMaps2.tikz}$}
	$}}_{\text{discarding maps}}	
\end{equation}
We refer to the discarding map obtained by doubling $\trace{\mathcal{E}}$ as the ``forest'' and to the discarding map obtained from the classical structure $\hbox{\input{symbols/ZdotSym.tex}}\!$ as the ``bridge''.  The scalars of \DDCategory{\fHilbCategory} are exactly the scalars $\reals^+$ of \CPMCategory{\fHilbCategory}, and hence the theory of density hypercubes is probabilistic. It is furthermore convex, because the following ``tree-on-a-bridge'' effects can be used to add-up maps of density hypercubes---analogously to the way ordinary discarding maps $\trace{\mathcal{H}}$ can be used to add-up CP maps in \CPMCategory{\fHilbCategory}---by expanding them in terms of the orthonormal bases $\ket{\psi_x}_{x \in X}$ associated \cite{coecke2013new} with the classical structures $\hbox{\input{symbols/ZbwdotSym.tex}}\!\!$:
\begin{equation}
	\scalebox{0.7}{$
		\input{pictures/DHclassicalDiscardingMaps2.tikz}
	$}
	\hspace{3mm} = \hspace{3mm}
	\scalebox{0.7}{$
		\input{pictures/DHquantumDiscardingMaps1.tikz}
	$}
	\hspace{3mm} = \hspace{3mm}
	\sum_{x \in X} \hspace{3mm}
	\scalebox{0.7}{$
		\input{pictures/DHclassicalDiscardingMaps3.tikz}
	$}
\end{equation}

\subsection{Component symmetries}

States in the theory \DDCategory{\fHilbCategory} take the form of fourth order tensors, an observation which prompted the choice of \inlineQuote{density hypercubes} as a name for the theory. If $(\ket{\psi_x})_{x \in X}$ is a choice of orthonormal basis for some finite-dimensional Hilbert space $H$, the states on $\DDCategory{H} = \mathcal{H} \otimes \mathcal{H}$ in \DDCategory{\fHilbCategory} can be expanded as follows in $\fHilbCategory$:
\begin{equation}
	\scalebox{0.8}{$\input{pictures/DHtensorForm1.tikz}$}
	\hspace{3mm} = \hspace{3mm}
	\begin{color}{gray}\sum_{x_{00},x_{01},x_{10},x_{11} \in X} \end{color} 
	\hspace{3mm}
	\scalebox{0.6}{$\input{pictures/DHtensorForm2.tikz}$}
\end{equation}
Recall that density matrices possess a $\integersMod{2}$ symmetry given by self-adjointness. This symmetry can be understood in terms of the following action $\tau: \integersMod{2} \rightarrow \Aut{\complexs}$ of $\integersMod{2}$ on the complex numbers:
\begin{equation}
	\begin{array}{rcrlcrcrl}
		\tau(0) &:=& z &\mapsto z &\hspace{2cm}& \tau(1) &:=& z &\mapsto z^\ast 
	\end{array}
\end{equation}
The components of a density matrix $\rho$ then satisfy the following equation, for every $a \in \integersMod{2}$ (trivial for $a=0$, self-adjoint for $a = 1$):
\begin{equation}
	\tau(a)( \rho_{\,x_0 \, x_1} ) = \rho_{x_{(0\oplus a)} \, x_{(1 \oplus a)}}
\end{equation}
Instead of a $\integersMod{2}$ symmetry, density hypercubes possess a $\integersMod{2} \times \integersMod{2}$ symmetry. This symmetry can be understood in terms of the following action $\tau: \integersMod{2} \times \integersMod{2} \rightarrow \Aut{\complexs}$ of $\integersMod{2} \times \integersMod{2}$ on the complex numbers:
\begin{equation}
	\begin{array}{rcrlcrcrl}
	  \tau(0,0) &:=& z &\mapsto z &\hspace{2cm} &\tau(0,1) &:=& z &\mapsto z^\ast \\
	  \tau(1,0) &:=& z &\mapsto z^\ast &\hspace{2cm} &\tau(1,1) &:=& z &\mapsto z 
	\end{array}
\end{equation}
The components of a density hypercube $\rho$ satisfy the following equation for every $(a,b) \in \integersMod{2} \times \integersMod{2}$, where by $\oplus$ we have denoted addition in $\integersMod{2}$:
\begin{equation}
	\tau(a,b)( \rho_{\,x_{(0,0)} \, x_{(0,1)} \, x_{(1,0)} \, x_{(1,1)}} ) = \rho_{\,x_{(0\oplus a,0\oplus b)} \, x_{(0\oplus a,1\oplus b)} \, x_{(1\oplus a,0\oplus b)} \, x_{(1\oplus a,1\oplus b)}}
\end{equation}
We see that the components are related by a trivial symmetry for $a = (0,0)$, by a self-adjoining symmetry for $a=(1,0)$ and $a=(0,1)$, and by a self-transposing symmetry in for $a=(1,1)$. An alternative way to look at this symmetry is observe that states of density hypercubes can all be expressed as certain sums of doubled states in the following form:
\begin{equation}
	\scalebox{0.8}{$\input{pictures/DHdoubledState.tikz}$}
\end{equation}
For these states, we have the usual self-conjugating $\integersMod{2}$ symmetry of density matrices $\Phi \otimes \overline{\Phi} \mapsto \Phi^\ast \otimes \overline{\Phi^\ast}$ as well as an independent self-transposing $\integersMod{2}$ symmetry $\Phi \otimes \overline{\Phi} \mapsto \overline{\Phi \otimes \overline{\Phi}}$, which taken together give the same $\integersMod{2} \times \integersMod{2}$ symmetry described above in terms of components.

In order to visualise the $\integersMod{2}\times\integersMod{2}$ symmetry action, we divide the components $\rho_{x_{00}x_{01}x_{10}x_{11}}$ of a $d$-dimensional density hypercube $\rho$ into 15 classes, depending on which indices $x_{00},x_{01},x_{10},x_{11}$ have same/distinct values chosen from the set $\{1,...,d\}$. We arrange the indices on a square: index $00$ is on the top left corner, $10$ acts as reflection about the vertical mid-line, $01$ acts as reflection about the horizontal mid-line and $11$ acts as 180\textsuperscript{o} rotation about the centre. We use colours as names for index values in $\{1,...,d\}$, with distinct colours denoting distinct values. 
\begin{equation}
	\underbrace{
		\fbox{\resizebox{!}{3mm}{\begin{tikzpicture}

	\begin{pgfonlayer}{nodelayer}
		\node[dot,fill=red,draw=red] (tl) at (-0.75,+0.75) {};
		\node[dot,fill=green,draw=green] (bl) at (-0.75,-0.75) {};
		\node[dot,fill=blue,draw=blue] (tr) at (+0.75,+0.75) {};
		\node[dot,fill=RoyalPurple,draw=RoyalPurple] (br) at (+0.75,-0.75) {};
	\end{pgfonlayer}
	\begin{pgfonlayer}{edgelayer}
	\end{pgfonlayer}

\end{tikzpicture}}}
	}_{\text{4 distinct}}
	\hspace{2mm}
	\underbrace{
		\fbox{\resizebox{!}{3mm}{\input{pictures/squares/twoEqualRight.tikz}}}
		\hspace{2mm}
		\fbox{\resizebox{!}{3mm}{\input{pictures/squares/twoEqualLeft.tikz}}}
		\hspace{2mm}
		\fbox{\resizebox{!}{3mm}{\input{pictures/squares/twoEqualTop.tikz}}}
		\hspace{2mm}
		\fbox{\resizebox{!}{3mm}{\input{pictures/squares/twoEqualBot.tikz}}}
		\hspace{2mm}
		\fbox{\resizebox{!}{3mm}{\input{pictures/squares/twoEqualBLTR.tikz}}}
		\hspace{2mm}
		\fbox{\resizebox{!}{3mm}{\input{pictures/squares/twoEqualTLBR.tikz}}}
	}_{\text{3 distinct}}
	\hspace{2mm}
	\underbrace{
		\fbox{\resizebox{!}{3mm}{\input{pictures/squares/oneDifferentBR.tikz}}}
		\hspace{2mm}
		\fbox{\resizebox{!}{3mm}{\input{pictures/squares/oneDifferentBL.tikz}}}
		\hspace{2mm}
		\fbox{\resizebox{!}{3mm}{\input{pictures/squares/oneDifferentTL.tikz}}}
		\hspace{2mm}
		\fbox{\resizebox{!}{3mm}{\input{pictures/squares/oneDifferentTR.tikz}}}
		\hspace{2mm}
		\fbox{\resizebox{!}{3mm}{\input{pictures/squares/twoAndTwoVertical.tikz}}}
		\hspace{2mm}
		\fbox{\resizebox{!}{3mm}{\input{pictures/squares/twoAndTwoHorizontal.tikz}}}
		\hspace{2mm}
		\fbox{\resizebox{!}{3mm}{\input{pictures/squares/twoAndTwoDiagonal.tikz}}}
	}_{\text{2 distinct}}
	\hspace{2mm}
	\underbrace{
		\fbox{\resizebox{!}{3mm}{\input{pictures/squares/allEqual.tikz}}}
	}_{\text{all equal}}
\end{equation}
For example, the component $\rho_{0321}$ of a $4^{+}$-dimensional system will fall into the 1\textsuperscript{st} class from the left above, the component $\rho_{0122}$ will fall into the 2\textsuperscript{nd} class, the component $\rho_{0003}$ into the 8\textsuperscript{th} class, the component $\rho_{0011}$ into the 12\textsuperscript{th} class and the component $\rho_{0000}$ into the 15\textsuperscript{th} class.

Then we look at the individual orbits of components in each class under the symmetry. Classes with components having orbits of order 4 are shown below: each orbit contributes a single independent complex value to the tensor, i.e. two independent real values, and each component class is annotated by the total number of independent real values contributed in dimension $d$. Just as we did above, we are using colours to denote values in $\{1,...,d\}$: the geometric action of $\integersMod{2}\times\integersMod{2}$ on the coloured vertices/edges of the squares exactly mirrors the algebraic action of $\integersMod{2}\times\integersMod{2}$ on the components in the different classes. 
\begin{equation}
	\underbrace{\addstackgap[4pt]{$
		\begin{array}{ccc}
		\resizebox{!}{3mm}{\begin{tikzpicture}

	\begin{pgfonlayer}{nodelayer}
		\node[dot,fill=red,draw=red] (tl) at (-0.75,+0.75) {};
		\node[dot,fill=green,draw=green] (bl) at (-0.75,-0.75) {};
		\node[dot,fill=blue,draw=blue] (tr) at (+0.75,+0.75) {};
		\node[dot,fill=RoyalPurple,draw=RoyalPurple] (br) at (+0.75,-0.75) {};
	\end{pgfonlayer}
	\begin{pgfonlayer}{edgelayer}
	\end{pgfonlayer}

\end{tikzpicture}}
		& \stackrel{10}{\leftrightarrow} &
		\resizebox{!}{3mm}{\begin{tikzpicture}

	\begin{pgfonlayer}{nodelayer}
		\node[dot,fill=blue,draw=blue] (tl) at (-0.75,+0.75) {};
		\node[dot,fill=RoyalPurple,draw=RoyalPurple] (bl) at (-0.75,-0.75) {};
		\node[dot,fill=red,draw=red] (tr) at (+0.75,+0.75) {};
		\node[dot,fill=green,draw=green] (br) at (+0.75,-0.75) {};
	\end{pgfonlayer}
	\begin{pgfonlayer}{edgelayer}
	\end{pgfonlayer}

\end{tikzpicture}}
		\\
		\raisebox{-2mm}{$\updownarrow \text{\scriptsize{01}}$} 
		& 
		\vspace{2mm} 
		\hspace{0mm}
		\raisebox{-2mm}{$
			\nearrow 
			\hspace{-4mm} \nwarrow 
			\hspace{-4mm} \searrow 
			\hspace{-3.65mm} \swarrow 
			\hspace{-1mm}\text{\scriptsize{11}}
			\hspace{-1mm}
		$}
		&
		\raisebox{-2mm}{$\updownarrow \text{\scriptsize{01}}$} 
		\\
		\resizebox{!}{3mm}{\begin{tikzpicture}

	\begin{pgfonlayer}{nodelayer}
		\node[dot,fill=green,draw=green] (tl) at (-0.75,+0.75) {};
		\node[dot,fill=red,draw=red] (bl) at (-0.75,-0.75) {};
		\node[dot,fill=RoyalPurple,draw=RoyalPurple] (tr) at (+0.75,+0.75) {};
		\node[dot,fill=blue,draw=blue] (br) at (+0.75,-0.75) {};
	\end{pgfonlayer}
	\begin{pgfonlayer}{edgelayer}
	\end{pgfonlayer}

\end{tikzpicture}}
		& \stackrel{10}{\leftrightarrow} &
		\resizebox{!}{3mm}{\begin{tikzpicture}

	\begin{pgfonlayer}{nodelayer}
		\node[dot,fill=RoyalPurple,draw=RoyalPurple] (tl) at (-0.75,+0.75) {};
		\node[dot,fill=blue,draw=blue] (bl) at (-0.75,-0.75) {};
		\node[dot,fill=green,draw=green] (tr) at (+0.75,+0.75) {};
		\node[dot,fill=red,draw=red] (br) at (+0.75,-0.75) {};
	\end{pgfonlayer}
	\begin{pgfonlayer}{edgelayer}
	\end{pgfonlayer}

\end{tikzpicture}}
		\\
		\end{array}
	$}}_{2\frac{1}{4}d(d-1)(d-2)(d-3)}
	\hspace{5mm}
	\underbrace{\addstackgap[4pt]{$
		\begin{array}{ccc}
		\resizebox{!}{3mm}{\input{pictures/squares/twoEqualRight.tikz}}
		& \stackrel{10}{\leftrightarrow} &
		\resizebox{!}{3mm}{\input{pictures/squares/twoEqualRight10.tikz}}
		\\
		\raisebox{-2mm}{$\updownarrow \text{\scriptsize{01}}$} 
		& 
		\vspace{2mm} 
		\hspace{0mm}
		\raisebox{-2mm}{$
			\nearrow 
			\hspace{-4mm} \nwarrow 
			\hspace{-4mm} \searrow 
			\hspace{-3.65mm} \swarrow 
			\hspace{-1mm}\text{\scriptsize{11}}
			\hspace{-1mm}
		$}
		&
		\raisebox{-2mm}{$\updownarrow \text{\scriptsize{01}}$} 
		\\
		\resizebox{!}{3mm}{\input{pictures/squares/twoEqualRight01.tikz}}
		& \stackrel{10}{\leftrightarrow} &
		\resizebox{!}{3mm}{\input{pictures/squares/twoEqualRight11.tikz}}
		\\
		\end{array}
	$}}_{2\frac{1}{4}d(d-1)(d-2)}
	\hspace{5mm}
	\underbrace{\addstackgap[4pt]{$
		\begin{array}{ccc}
		\resizebox{!}{3mm}{\input{pictures/squares/twoEqualTop.tikz}}
		& \stackrel{10}{\leftrightarrow} &
		\resizebox{!}{3mm}{\input{pictures/squares/twoEqualTop10.tikz}}
		\\
		\raisebox{-2mm}{$\updownarrow \text{\scriptsize{01}}$} 
		& 
		\vspace{2mm} 
		\hspace{0mm}
		\raisebox{-2mm}{$
			\nearrow 
			\hspace{-4mm} \nwarrow 
			\hspace{-4mm} \searrow 
			\hspace{-3.65mm} \swarrow 
			\hspace{-1mm}\text{\scriptsize{11}}
			\hspace{-1mm}
		$}
		&
		\raisebox{-2mm}{$\updownarrow \text{\scriptsize{01}}$} 
		\\
		\resizebox{!}{3mm}{\input{pictures/squares/twoEqualTop01.tikz}}
		& \stackrel{10}{\leftrightarrow} &
		\resizebox{!}{3mm}{\input{pictures/squares/twoEqualTop11.tikz}}
		\\
		\end{array}
	$}}_{2\frac{1}{4}d(d-1)(d-2)}
	\hspace{5mm}
	\underbrace{\addstackgap[4pt]{$
		\begin{array}{ccc}
		\resizebox{!}{3mm}{\input{pictures/squares/twoEqualTLBR.tikz}}
		& \stackrel{10}{\leftrightarrow} &
		\resizebox{!}{3mm}{\input{pictures/squares/twoEqualTLBR10.tikz}}
		\\
		\raisebox{-2mm}{$\updownarrow \text{\scriptsize{01}}$} 
		& 
		\vspace{2mm} 
		\hspace{0mm}
		\raisebox{-2mm}{$
			\nearrow 
			\hspace{-4mm} \nwarrow 
			\hspace{-4mm} \searrow 
			\hspace{-3.65mm} \swarrow 
			\hspace{-1mm}\text{\scriptsize{11}}
			\hspace{-1mm}
		$}
		&
		\raisebox{-2mm}{$\updownarrow \text{\scriptsize{01}}$} 
		\\
		\resizebox{!}{3mm}{\input{pictures/squares/twoEqualTLBR01.tikz}}
		& \stackrel{10}{\leftrightarrow} &
		\resizebox{!}{3mm}{\input{pictures/squares/twoEqualTLBR11.tikz}}
		\\
		\end{array}
	$}}_{2\frac{1}{4}d(d-1)(d-2)}
	\hspace{5mm}
	\underbrace{\addstackgap[4pt]{$
		\begin{array}{ccc}
		\resizebox{!}{3mm}{\input{pictures/squares/oneDifferentBR.tikz}}
		& \stackrel{10}{\leftrightarrow} &
		\resizebox{!}{3mm}{\input{pictures/squares/oneDifferentBR10.tikz}}
		\\
		\raisebox{-2mm}{$\updownarrow \text{\scriptsize{01}}$} 
		& 
		\vspace{2mm} 
		\hspace{0mm}
		\raisebox{-2mm}{$
			\nearrow 
			\hspace{-4mm} \nwarrow 
			\hspace{-4mm} \searrow 
			\hspace{-3.65mm} \swarrow 
			\hspace{-1mm}\text{\scriptsize{11}}
			\hspace{-1mm}
		$}
		&
		\raisebox{-2mm}{$\updownarrow \text{\scriptsize{01}}$} 
		\\
		\resizebox{!}{3mm}{\input{pictures/squares/oneDifferentBR01.tikz}}
		& \stackrel{10}{\leftrightarrow} &
		\resizebox{!}{3mm}{\input{pictures/squares/oneDifferentBR11.tikz}}
		\\
		\end{array}
	$}}_{2\frac{1}{4}d(d-1)}
\end{equation}
Classes with components having orbits of order 2 and 1 are shown below, each component class annotated by the total number of independent real values contributed in dimension $d$. Each orbit in the first, second and fourth classes contributes a single independent real value, because each component is stabilised by (at least) one self-adjoining symmetry; each orbit in the third class contributes instead two independent real values, because the components are only stabilised by a self-transposing symmetry.
\begin{equation}
	\underbrace{\addstackgap[4pt]{$
		\begin{array}{ccc}
		\resizebox{!}{3mm}{\input{pictures/squares/twoAndTwoVertical.tikz}}
		& \stackrel{10,11}{\leftrightarrow} &
		\resizebox{!}{3mm}{\input{pictures/squares/twoAndTwoVertical10.tikz}}
		\end{array}
	$}}_{\frac{1}{2}d(d-1)}
	\hspace{5mm}
	\underbrace{\addstackgap[4pt]{$
		\begin{array}{ccc}
		\resizebox{!}{3mm}{\input{pictures/squares/twoAndTwoHorizontal.tikz}}
		& \stackrel{01,11}{\leftrightarrow} &
		\resizebox{!}{3mm}{\input{pictures/squares/twoAndTwoHorizontal01.tikz}}
		\end{array}
	$}}_{\frac{1}{2}d(d-1)}
	\hspace{5mm}
	\underbrace{\addstackgap[4pt]{$
		\begin{array}{ccc}
		\resizebox{!}{3mm}{\input{pictures/squares/twoAndTwoDiagonal.tikz}}
		& \stackrel{10,01}{\leftrightarrow} &
		\resizebox{!}{3mm}{\input{pictures/squares/twoAndTwoDiagonal10.tikz}}
		\end{array}
	$}}_{2\frac{1}{2}d(d-1)}
	\hspace{5mm}
	\underbrace{\addstackgap[4pt]{$
		\resizebox{!}{3mm}{\input{pictures/squares/allEqual.tikz}}
	$}}_{\text{d}}
\end{equation}
Adding up the contributions from all orbit classes, we see that the states of $d$-dimensional density hypercubes form a convex cone of real dimension $\frac{1}{2}(d^4-3d^3+7d^2-3d)$ within the $(2d^4)$-dimensional real vector space of complex fourth-order tensors.

\subsection{Normalisation and causality}

The ``forest'' discarding maps $\trace{\,\,\DDCategory{H}}:=\CPMCategory{\trace{\mathcal{H}}}$ in \DDCategory{\fHilbCategory} (i.e.\ the doubled versions of the discarding maps of \CPMCategory{\fHilbCategory}) form an environment structure \cite{gogioso2017categorical,coecke2010environment}, and we say that a map of density hypercubes is \defi{normalised} if the corresponding CP map is trace preserving (with normalised states as a special case): 
\begin{equation}
	\scalebox{0.7}{$\input{pictures/normalisedDHmap1.tikz}$}
	\hspace{3mm} \text{normalised} \hspace{3mm}
	\Leftrightarrow \hspace{3mm}
	\scalebox{0.7}{$\input{pictures/normalisedDHmap2.tikz}$}
	\hspace{3mm} = \hspace{3mm}
	\scalebox{0.7}{$\begin{tikzpicture}

	\begin{pgfonlayer}{nodelayer}
		\node[style=none] (Hb) at (-2,+2) {};
		\node[style=none] (H) at (-2,-2) {};
		\node[style=trace] (traceHb) at (+2,+2) {};
		\node[style=trace] (traceH) at (+2,-2) {};
	\end{pgfonlayer}
	\begin{pgfonlayer}{edgelayer}
		\draw[-] (Hb.center) to (traceHb);
		\draw[-] (H.center) to (traceH);
	\end{pgfonlayer}

\end{tikzpicture}$}
\end{equation}
Normalised maps of density hypercubes form a sub-SMC of \DDCategory{\fHilbCategory}, which we refer to as the \defi{normalised sub-category}. \defi{Sub-normalised} maps of density hypercubes can be defined analogously by requiring the corresponding CP map to be trace non-increasing: they also form a sub-SMC of \DDCategory{\fHilbCategory}, which we refer to as the \defi{sub-normalised sub-category}.

Despite the presence of several kinds of discarding maps, the following results shows that the sub-normalised sub-category is causal \cite{Chiribella-purification}, or equivalently that that the normalised sub-category is terminal \cite{coecke2013causal,coecke2016terminality}.

\newcounter{proposition_causality_c}
\setcounter{proposition_causality_c}{\value{theorem_c}}
\begin{proposition}
\label{proposition_causality}
	The process theory $\DDCategory{\fHilbCategory}$ is causal, in the following sense: for every object $\DDCategory{H}$, the only effect $\DDCategory{H} \rightarrow \reals^+$ in \DDCategory{\fHilbCategory} which yields the scalar $1$ on all normalised states of $\DDCategory{H}$ is the ``forest'' discarding map of density hypercubes $\trace{\,\,\DDCategory{H}}$. 
\end{proposition}

\section{Decoherence and Hyper-decoherence}
\label{section_hyperDecoherence}

So far, we have constructed a symmetric monoidal category, which is enriched in convex cones and comes equipped with an environment structure providing a notion of normalization. The final ingredients necessary for the definition of the \defi{categorical probabilistic theory of density hypercubes} is the demonstration that classical systems and quantum systems arise in the Karoubi envelope of \DDCategory{\fHilbCategory} by choosing some suitable family of decoherence and hyper-decoherence maps.

\subsection{Decoherence to classical theory}

Consider a finite-dimensional Hilbert space $H$ and a classical structure $\hbox{\input{symbols/ZdotSym.tex}}\!$ on it, associated with some orthonormal basis $(\ket{\psi_x})_{x \in X}$. We define the \defi{$\hbox{\input{symbols/ZdotSym.tex}}\!$-decoherence map} $\decoh{\hbox{\input{symbols/ZdotSym.tex}}\!}$ on the density hypercube $\DDCategory{H}$ to be the following morphism in \DDCategory{\fHilbCategory}:
\begin{equation}
	\decoh{\hbox{\input{symbols/ZdotSym.tex}}\!}
	\hspace{3mm} := \hspace{3mm}
	\scalebox{0.7}{$
		\input{pictures/DHdecoherence1.tikz}
	$}
	\hspace{3mm} = \hspace{3mm}
	\sum_{x \in X} \hspace{3mm}
	\scalebox{0.7}{$
		\input{pictures/DHdecoherence2.tikz}
	$}
\end{equation}
The $\decoh{\hbox{\input{symbols/ZdotSym.tex}}\!}$ map defined above is idempotent, so it can be used to define classical systems via the Karoubi envelope construction---in the same way as ordinary decoherence maps gives rise to classical systems in quantum theory. It should be noted that decoherence maps defined this way are sub-normalised but not normalised, so that the hyperquantum-to-classical transition in the theory of density hypercubes is not deterministic; we defer further discussion of this point to the next sub-section on hyper-decoherence.

\newcounter{proposition_classical_c}
\setcounter{proposition_classical_c}{\value{theorem_c}}
\begin{proposition}
\label{proposition_classical}
	Let $\KaroubiEnvelope{\DDCategory{\fHilbCategory}}$ be the Karoubi envelope of \DDCategory{\fHilbCategory}, and write $\KaroubiEnvelope{\DDCategory{\fHilbCategory}}_K$ for the full subcategory of $\KaroubiEnvelope{\DDCategory{\fHilbCategory}}$ spanned by objects in the form $(\DDCategory{H},\decoh{\hbox{\input{symbols/ZdotSym.tex}}\!})$. There is an $\reals^+$-linear monoidal equivalence of categories between $\KaroubiEnvelope{\DDCategory{\fHilbCategory}}_K$ and the probabilistic theory $\RMatCategory{\reals^+}$ of classical systems. Furthermore, classical stochastic maps correspond to the maps in $\KaroubiEnvelope{\DDCategory{\fHilbCategory}}_K$ normalised with respect to the discarding maps $\trace{\,\,(\DDCategory{H},\decoh{\hbox{\input{symbols/ZdotSym.tex}}\!})} := \trace{\,\,\DDCategory{H}} \circ \decoh{\hbox{\input{symbols/ZdotSym.tex}}\!}$, which we can write explicitly as follows:
	\newcounter{proposition_classical_c_eq}
	\setcounter{proposition_classical_c_eq}{\value{equation}}
	\begin{equation}
		\label{proposition_classical_eq_label}
		\trace{\,\,(\DDCategory{H},\decoh{\hbox{\input{symbols/ZdotSym.tex}}\!})}
		\hspace{3mm} := \hspace{3mm}
		\scalebox{0.7}{$
			\input{pictures/DHclassicalDiscardingMaps1.tikz}
		$}
		\hspace{3mm} = \hspace{3mm}
		\scalebox{0.7}{$
			\input{pictures/DHclassicalDiscardingMaps2.tikz}
		$}
	\end{equation}
\end{proposition}

\subsection{Hyper-decoherence to quantum theory}

We now show that the quantum systems arise in the Karoubi envelope as well, via suitable \defi{hyper-decoherence} maps. Recall that the generic discarding map in the theory of density hypercubes involved two pieces: (the doubled version of) a traditional discarding map from \CPMCategory{\fHilbCategory} and a second ``tree-on-a-bridge'' discarding map derived from a classical structure $\hbox{\input{symbols/ZdotSym.tex}}\!$. In the previous sub-section, we saw that the latter is the discarding map of some classical system living in the Karoubi envelope \KaroubiEnvelope{\DDCategory{\fHilbCategory}}, and that it can be used to define the ``hyper-quantum--to--classical'' decoherence maps. In this sub-section, we shall see that this ``hyper-quantum--to--classical'' decoherence process can be understood in two steps: a ``hyper-quantum--to--quantum'' hyper-decoherence, followed by the usual ``quantum--to--classical'' decoherence.

If $\hbox{\input{symbols/ZdotSym.tex}}\!$ is a classical structure on a density hypercube $\DDCategory{H}$, we define the \defi{$\hbox{\input{symbols/ZdotSym.tex}}\!$-hyper-decoherence map} $\hypdecoh{\hbox{\input{symbols/ZdotSym.tex}}\!}$ to be the following map of density hypercubes:
\begin{equation}	
	\hypdecoh{\hbox{\input{symbols/ZdotSym.tex}}\!}
	\hspace{3mm} := \hspace{3mm}
	\scalebox{0.7}{$
		\input{pictures/DHhyperdecoherence.tikz}
	$}
\end{equation}
Hyper-decoherence maps are idempotent, and hence we can consider the full subcategory $\mathcal{C}$ of the Karoubi envelope \KaroubiEnvelope{\DDCategory{\fHilbCategory}} spanned by objects in the form $(\DDCategory{H},\hypdecoh{\hbox{\input{symbols/ZdotSym.tex}}\!})$: doing so allows us to prove that the hyper-decoherence maps defined above truly provide the desired ``hyper-quantum--to--quantum'' decoherence, as considered by \cite{Lee-Selby-interference,lee2017no}.

\newcounter{proposition_quantum_c}
\setcounter{proposition_quantum_c}{\value{theorem_c}}
\begin{proposition}
\label{proposition_quantum}
	Let $\KaroubiEnvelope{\DDCategory{\fHilbCategory}}$ be the Karoubi envelope of \DDCategory{\fHilbCategory}, and write $\KaroubiEnvelope{\DDCategory{\fHilbCategory}}_Q$ for the full subcategory of $\KaroubiEnvelope{\DDCategory{\fHilbCategory}}$ spanned by objects in the form $(\DDCategory{H},\hypdecoh{\hbox{\input{symbols/ZdotSym.tex}}\!})$. There is an $\reals^+$-linear monoidal equivalence of categories between $\KaroubiEnvelope{\DDCategory{\fHilbCategory}}_Q$ and the probabilistic theory $\CPMCategory{\fHilbCategory}$ of quantum systems and CP maps between them. Furthermore, trace-preserving CP maps correspond to the maps in $\KaroubiEnvelope{\DDCategory{\fHilbCategory}}_Q$ normalised with respect to the discarding maps $\trace{\,\,(\DDCategory{H},\hypdecoh{\hbox{\input{symbols/ZdotSym.tex}}\!})} := \trace{\,\,\DDCategory{H}} \circ \hypdecoh{\hbox{\input{symbols/ZdotSym.tex}}\!}$, which we can write explicitly as follows:
	\newcounter{proposition_quantum_c_eq}
	\setcounter{proposition_quantum_c_eq}{\value{equation}}
	\begin{equation}
		\label{proposition_quantum_eq_label}
		\trace{\,\,(\DDCategory{H},\hypdecoh{\hbox{\input{symbols/ZdotSym.tex}}\!})}
		\hspace{3mm} := \hspace{3mm}
		\scalebox{0.7}{$
			\input{pictures/DHquantumDiscardingMaps1.tikz}
		$}
		\hspace{3mm} = \hspace{3mm}
		\scalebox{0.7}{$
			\input{pictures/DHclassicalDiscardingMaps2.tikz}
		$}
	\end{equation}
\end{proposition} 
Taking the double-dilation construction together with the content of Propositions \ref{proposition_classical} and \ref{proposition_quantum}, we come to the following definition of a categorical probabilistic theory \cite{gogioso2017categorical} of density hypercubes.
\begin{definition}  
The \defi{categorical probabilistic theory of density hypercubes} \DHCategory{\fHilbCategory} is defined the be the full sub-SMC of \KaroubiEnvelope{\DDCategory{\fHilbCategory}} spanned by objects in the following form: 
\begin{itemize}
	\item the \defi{density hypercubes} $(\DDCategory{H},\id{\DDCategory{H}})$;
	\item the \defi{quantum systems} $(\DDCategory{H},\hypdecoh{\hbox{\input{symbols/ZdotSym.tex}}\!})$, for all classical structures $\hbox{\input{symbols/ZdotSym.tex}}\!$ on $H$;
	\item the \defi{classical systems} $(\DDCategory{H},\decoh{\hbox{\input{symbols/ZdotSym.tex}}\!})$, for all classical structures $\hbox{\input{symbols/ZdotSym.tex}}\!$ on $H$.
\end{itemize}
The environment structure for the categorical probabilistic theory is given by the discarding maps $\trace{\,\,\DDCategory{H}}$, $\trace{\,\,(\DDCategory{H},\hypdecoh{\hbox{\input{symbols/ZdotSym.tex}}\!})}$ and $\trace{\,\,(\DDCategory{H},\decoh{\hbox{\input{symbols/ZdotSym.tex}}\!})}$ respectively. The classical sub-category for the categorical probabilistic theory is the full sub-SMC spanned by the classical systems.
\end{definition}

The hyper-quantum--to--classical and hyper-quantum--to--quantum decoherence maps of density hypercubes play well together with the quantum--to--classical decoherence map of quantum theory: the decoherence map $\decoh{\hbox{\input{symbols/ZdotSym.tex}}\!}:(\DDCategory{H},\id{\DDCategory{H}}) \rightarrow (\DDCategory{H},\decoh{\hbox{\input{symbols/ZdotSym.tex}}\!})$ of density hypercubes factors, as one would expect, into the hyper-decoherence map $\hypdecoh{\hbox{\input{symbols/ZdotSym.tex}}\!}:(\DDCategory{H},\id{\DDCategory{H}}) \rightarrow (\DDCategory{H},\hypdecoh{\hbox{\input{symbols/ZdotSym.tex}}\!})$ followed by the decoherence map $\decoh{\hbox{\input{symbols/ZdotSym.tex}}\!}:(\DDCategory{H},\hypdecoh{\hbox{\input{symbols/ZdotSym.tex}}\!}) \rightarrow (\DDCategory{H},\decoh{\hbox{\input{symbols/ZdotSym.tex}}\!})$ of quantum systems. From this, it is clear that the reason why hyper-quantum--to--classical transition was sub-normalised is that the hyper-quantum--to--quantum transition itself is sub-normalised (see Appendix \ref{appendix_extension}).

The sub-normalisation of hyper-decoherence maps is a sign that the theory of density hypercubes presented here is still partially incomplete, and that some suitable extension will need to be researched in the future. What we know for sure is that the current theory does not satisfy the no-restriction condition on effects, and that an extension in which hyper-decoherence maps are normalised is possible: the additional effect needed by normalisation exists in \CPMCategory{\fHilbCategory} and is non-negative on all states of \DDCategory{\fHilbCategory} (see Appendix \ref{appendix_extension}). In line with the recent no-go theorem of \cite{lee2017no}, preliminary considerations seem to indicated that the addition of said effect would mean that the theory no longer satisfies purification.

\newcommand{\redbullet}{\begin{color}{red}\bullet\end{color}}
\newcommand{\bluebullet}{\begin{color}{blue}\bullet\end{color}}
\newcommand{\greenbullet}{\begin{color}{green}\bullet\end{color}}
\newcommand{\RoyalPurplebullet}{\begin{color}{RoyalPurple}\bullet\end{color}}

\newpage
\section{Higher Order Interference}
\label{section_higherOrderInterference}

In this section, we will show that the theory of density hypercubes displays third- and fourth-order interference effects, broadly inspired by the framework for higher-order interference in GPTs presented by \cite{HOP,Lee-Selby-interference,Barnum-interference}. Because interference has to do with decompositions of the identity map in terms of certain projectors, we begin by introducing a handy graphical notation for keeping track of the various pieces that the identity map is composed of.

The identity map of hyper-quantum systems $\id{\DDCategory{H}} : \DDCategory{H} \rightarrow \DDCategory{H}$ takes the following explicit form in $\fHilbCategory$, for any orthonormal basis $(\ket{\psi_x})_{x \in X}$ of the Hilbert space ${H}$:
\begin{equation}
	\scalebox{0.8}{$
		\input{pictures/DHidentity1.tikz}
	$}
		\hspace{3mm} = \hspace{1mm}
	\begin{color}{gray}
		\sum_{x_{00},x_{01},x_{10},x_{11}\in X} \hspace{2mm}
	\end{color}
	\scalebox{0.8}{$
		\input{pictures/DHidentity2.tikz}
	$}
\end{equation}
In order to denote the pieces in the decomposition corresponding to specific values $x_{00}, x_{01}, x_{10}, x_{11} \in X$ of the indices, we adopt the following graphical notation, inspired by the $\integersMod{2} \times \integersMod{2}$ symmetry of the components:
\begin{equation}
	\scalebox{0.8}{$
		\input{pictures/DHprojectorNotation1old.tikz}
	$}
		\hspace{3mm} := \hspace{1mm}
	\scalebox{0.8}{$
		\input{pictures/DHidentity2.tikz}
	$}
\end{equation}
In fact, we will adopt the same colour-based notation for index values which we originally introduced in Section \ref{section_densityHypercubes}, so that the following is a decomposition piece involving two distinct index values $\{\redbullet,\bluebullet\} \subseteq X$:
\begin{equation}
	\scalebox{1}{$
		\input{pictures/DHprojectorNotation2.tikz}
	$}
	\hspace{3mm} := \hspace{3mm}
	\scalebox{1}{$
		\input{pictures/DHidentity2example.tikz}
	$}
\end{equation}
Using the colour-based notation defined above for its pieces, the identity on a 2-dimensional hyper-quantum system (with $X = \{\bluebullet,\redbullet\}$) would be fully decomposed as follows:
\begin{equation}
		\id{\complexs^2}
		\hspace{1mm} = \hspace{1mm}
		{\resizebox{!}{3mm}{\input{pictures/squares/allEqual.tikz}}}
		+
		{\resizebox{!}{3mm}{\input{pictures/squares/allEqualred.tikz}}}
		+
		{\resizebox{!}{3mm}{\input{pictures/squares/twoAndTwoVertical.tikz}}}
		+
		{\resizebox{!}{3mm}{\input{pictures/squares/twoAndTwoVertical10.tikz}}}
		+
		{\resizebox{!}{3mm}{\input{pictures/squares/twoAndTwoHorizontal.tikz}}}
		+
		{\resizebox{!}{3mm}{\input{pictures/squares/twoAndTwoHorizontal01.tikz}}}
		+
		{\resizebox{!}{3mm}{\input{pictures/squares/twoAndTwoDiagonal.tikz}}}
		+
		{\resizebox{!}{3mm}{\input{pictures/squares/twoAndTwoDiagonal10.tikz}}}
		+
		{\resizebox{!}{3mm}{\input{pictures/squares/oneDifferentBR.tikz}}}
		+
		{\resizebox{!}{3mm}{\input{pictures/squares/oneDifferentBRred.tikz}}}
		+
		{\resizebox{!}{3mm}{\input{pictures/squares/oneDifferentBL.tikz}}}
		+
		{\resizebox{!}{3mm}{\input{pictures/squares/oneDifferentBLred.tikz}}}
		+
		{\resizebox{!}{3mm}{\input{pictures/squares/oneDifferentTL.tikz}}}
		+
		{\resizebox{!}{3mm}{\input{pictures/squares/oneDifferentTLred.tikz}}}
		+
		{\resizebox{!}{3mm}{\input{pictures/squares/oneDifferentTR.tikz}}}
		+
		{\resizebox{!}{3mm}{\input{pictures/squares/oneDifferentTRred.tikz}}}
\end{equation}
The same notation can be used to graphically decompose projectors corresponding to various subspaces determined by the orthonormal basis $(\ket{\psi_x})_{x \in X}$. For any non-empty subset $U \subseteq X$, we define the following projector on $\DDCategory{H}$:
\begin{equation}
	P_{U} := \DDCategorySym\left(\sum_{x \in U} \ket{\psi_x}\bra{\psi_x}\right)
\end{equation}
In particular, the $P_{\{\bluebullet\}}$ for $\bluebullet \in X$ are the projectors corresponding to the individual vectors $\ket{\psi_{\bluebullet}}$ of the basis, while $P_{X}$ is the identity $\id{\DDCategory{H}}$. No matter how large $X$ is (with $\#X \geq 2$), the projectors $P_{\{\bluebullet,\redbullet\}}$ corresponding to 2-element subsets $\{\bluebullet,\redbullet\} \subseteq X$ are always decomposed as follows:
\begin{equation}
		P_{\{\bluebullet,\redbullet\}}
		\hspace{1mm} = \hspace{1mm}
		{\resizebox{!}{3mm}{\input{pictures/squares/allEqual.tikz}}}
		+
		{\resizebox{!}{3mm}{\input{pictures/squares/allEqualred.tikz}}}
		+
		{\resizebox{!}{3mm}{\input{pictures/squares/twoAndTwoVertical.tikz}}}
		+
		{\resizebox{!}{3mm}{\input{pictures/squares/twoAndTwoVertical10.tikz}}}
		+
		{\resizebox{!}{3mm}{\input{pictures/squares/twoAndTwoHorizontal.tikz}}}
		+
		{\resizebox{!}{3mm}{\input{pictures/squares/twoAndTwoHorizontal01.tikz}}}
		+
		{\resizebox{!}{3mm}{\input{pictures/squares/twoAndTwoDiagonal.tikz}}}
		+
		{\resizebox{!}{3mm}{\input{pictures/squares/twoAndTwoDiagonal10.tikz}}}
		+
		{\resizebox{!}{3mm}{\input{pictures/squares/oneDifferentBR.tikz}}}
		+
		{\resizebox{!}{3mm}{\input{pictures/squares/oneDifferentBRred.tikz}}}
		+
		{\resizebox{!}{3mm}{\input{pictures/squares/oneDifferentBL.tikz}}}
		+
		{\resizebox{!}{3mm}{\input{pictures/squares/oneDifferentBLred.tikz}}}
		+
		{\resizebox{!}{3mm}{\input{pictures/squares/oneDifferentTL.tikz}}}
		+
		{\resizebox{!}{3mm}{\input{pictures/squares/oneDifferentTLred.tikz}}}
		+
		{\resizebox{!}{3mm}{\input{pictures/squares/oneDifferentTR.tikz}}}
		+
		{\resizebox{!}{3mm}{\input{pictures/squares/oneDifferentTRred.tikz}}}
\end{equation}
The presence of higher order interference in the theory of density hypercubes is really a matter of shapes: when the dimension of $\mathcal{H}$ is at least 3, the identity contains pieces of shapes which do not appear in projectors for 1-element and 2-element subsets. Because of this, in the theory of density hypercubes the probabilities obtained from 1-slit and 2-slit interference experiments will not be enough to explain the probabilities obtained from 3-slit and/or 4-slit experiments; however, the probabilities obtained from 1-slit, 2-slit, 3-slit and 4-slit experiments will always be enough to explain the probabilities obtained in experiments with 5 or more slits. 

Below you can see an atlas of all possible shapes that pieces of the identity can take in our graphical notation, together with a note of the smallest dimension that a projector must have to contain pieces of that shape:
\begin{equation}
	\underbrace{
		\fbox{\resizebox{!}{3mm}{\input{pictures/squares/shapes/allEqual.tikz}}}
	}_{\text{1-dim}}
	\hspace{2mm}
	\underbrace{
		\fbox{\resizebox{!}{3mm}{\input{pictures/squares/shapes/oneDifferentBR.tikz}}}
		\hspace{2mm}
		\fbox{\resizebox{!}{3mm}{\input{pictures/squares/shapes/oneDifferentBL.tikz}}}
		\hspace{2mm}
		\fbox{\resizebox{!}{3mm}{\input{pictures/squares/shapes/oneDifferentTL.tikz}}}
		\hspace{2mm}
		\fbox{\resizebox{!}{3mm}{\input{pictures/squares/shapes/oneDifferentTR.tikz}}}
		\hspace{2mm}
		\fbox{\resizebox{!}{3mm}{\input{pictures/squares/shapes/twoAndTwoVertical.tikz}}}
		\hspace{2mm}
		\fbox{\resizebox{!}{3mm}{\input{pictures/squares/shapes/twoAndTwoHorizontal.tikz}}}
		\hspace{2mm}
		\fbox{\resizebox{!}{3mm}{\input{pictures/squares/shapes/twoAndTwoDiagonal.tikz}}}
	}_{\text{2-dim}}
	\hspace{2mm}
	\underbrace{
		\fbox{\resizebox{!}{3mm}{\input{pictures/squares/shapes/twoEqualRight.tikz}}}
		\hspace{2mm}
		\fbox{\resizebox{!}{3mm}{\input{pictures/squares/shapes/twoEqualLeft.tikz}}}
		\hspace{2mm}
		\fbox{\resizebox{!}{3mm}{\input{pictures/squares/shapes/twoEqualTop.tikz}}}
		\hspace{2mm}
		\fbox{\resizebox{!}{3mm}{\input{pictures/squares/shapes/twoEqualBot.tikz}}}
		\hspace{2mm}
		\fbox{\resizebox{!}{3mm}{\input{pictures/squares/shapes/twoEqualBLTR.tikz}}}
		\hspace{2mm}
		\fbox{\resizebox{!}{3mm}{\input{pictures/squares/shapes/twoEqualTLBR.tikz}}}
	}_{\text{3-dim}}
	\hspace{2mm}
	\underbrace{
		\fbox{\resizebox{!}{3mm}{\begin{tikzpicture}

	\begin{pgfonlayer}{nodelayer}
		\node[dot,fill=black,draw=black] (tl) at (-0.75,+0.75) {};
		\node[dot,fill=black,draw=black] (bl) at (-0.75,-0.75) {};
		\node[dot,fill=black,draw=black] (tr) at (+0.75,+0.75) {};
		\node[dot,fill=black,draw=black] (br) at (+0.75,-0.75) {};
	\end{pgfonlayer}
	\begin{pgfonlayer}{edgelayer}
	\end{pgfonlayer}

\end{tikzpicture}}}
	}_{\text{4-dim}}
\end{equation}
The shape labelled as 1-dimensional only requires a single index value, and hence pieces of that shape appear in all projectors. The shapes labelled as 2-dimensional all require exactly two distinct index values, and hence pieces of those shapes can only appear in projectors for subsets with at least 2 elements. The shapes labelled as 3-dimensional all require exactly three distinct index values, and hence pieces of those shapes can only appear in projectors for subsets with at least 3 elements. Finally, the shape labelled as 4-dimensional requires exactly four index values, and hence pieces of that shape can only appear in projectors for subsets with at least 4 elements.

Thanks to the graphical notation introduced above, we already have a first intuition of why density hypercubes display higher-order interference. However, a rigorous proof requires a complete set-up with states, projectors, measurements and probabilities for a $d$-slit interference experiment, so that is what we now endeavour to provide. 
\begin{enumerate}
	\item We choose a $d$-dimensional space $H \isom \complexs^d$, and we value our tensor indices in the set $X = \{1,...,d\}$ (the same set that we use to label the $d$ slits).
	\item We fix an orthonormal basis $(\ket{x})_{x \in X}$, and we interpret $\ket{x}$ to be the state in which the particle goes through slit $x$ with certainty.
	\item The initial state for the particle is the superposition state in which the particle goes through each slit with the same amplitude. More precisely, it is the pure normalised density hypercube state $\rho_+$ corresponding to the vector $\frac{1}{\sqrt{d}}\ket{\psi_+} := \frac{1}{\sqrt{d}}(\ket{1} + ... + \ket{d})$:
	\begin{equation}
		\rho_+
		\hspace{1mm} := \hspace{1mm}
		\frac{1}{d^2} \hspace{1mm}
		\scalebox{0.8}{$
			\input{pictures/InterferenceExpInitialState.tikz}
		$}
	\end{equation}
	\item The particle goes through some non-empty subset $U \subseteq X$ of slits at random: afterwards, the experimenter knows which subset the particle passed through, but no more information than that is available in the universe.
	\item The particle is measured at the screen, and the experimenter estimates the probability $\mathbb{P}[+|U]$ that the particle is still in state $\rho_+$ after having passed through the given subset $U$ of the slits:
	\begin{equation}
		\mathbb{P}[+|U] 
		\hspace{1mm} := \hspace{1mm}
		\frac{1}{d^2} \hspace{1mm}
		\scalebox{0.8}{$
			\input{pictures/InterferenceExpProbability.tikz}
		$}
		\hspace{1mm} \frac{1}{d^2} 
	\end{equation}
\end{enumerate}
It is immediate to see that the outcome probability $\mathbb{P}[+|U]$ depends solely on the number of different pieces appearing in the decomposition of the projector $P_U$: 
\begin{equation}
	\mathbb{P}[+|U] = \frac{1}{d^4} \cdot \textnormal{number of pieces in }P_U
\end{equation}
To count the number of pieces in $P_U$, it is convenient to group them by shapes. If $U$ is a subset of size $k$, standard combinatorial arguments can be used to obtain the number of pieces of each shape appearing in the decomposition (as a convention, we set ${k \choose{j}} = 0$ for $j > k$):
\begin{equation}
	\underbrace{
		\fbox{\resizebox{!}{3mm}{\input{pictures/squares/shapes/allEqual.tikz}}}
	}_{{k \choose{1}} \cdot 1!}
	\hspace{2mm}
	\underbrace{
		\fbox{\resizebox{!}{3mm}{\input{pictures/squares/shapes/oneDifferentBR.tikz}}}
		\hspace{2mm}
		\fbox{\resizebox{!}{3mm}{\input{pictures/squares/shapes/oneDifferentBL.tikz}}}
		\hspace{2mm}
		\fbox{\resizebox{!}{3mm}{\input{pictures/squares/shapes/oneDifferentTL.tikz}}}
		\hspace{2mm}
		\fbox{\resizebox{!}{3mm}{\input{pictures/squares/shapes/oneDifferentTR.tikz}}}
		\hspace{2mm}
		\fbox{\resizebox{!}{3mm}{\input{pictures/squares/shapes/twoAndTwoVertical.tikz}}}
		\hspace{2mm}
		\fbox{\resizebox{!}{3mm}{\input{pictures/squares/shapes/twoAndTwoHorizontal.tikz}}}
		\hspace{2mm}
		\fbox{\resizebox{!}{3mm}{\input{pictures/squares/shapes/twoAndTwoDiagonal.tikz}}}
	}_{\text{7 shapes, }{k \choose{2}} \cdot 2!\text{ each}}
	\hspace{2mm}
	\underbrace{
		\fbox{\resizebox{!}{3mm}{\input{pictures/squares/shapes/twoEqualRight.tikz}}}
		\hspace{2mm}
		\fbox{\resizebox{!}{3mm}{\input{pictures/squares/shapes/twoEqualLeft.tikz}}}
		\hspace{2mm}
		\fbox{\resizebox{!}{3mm}{\input{pictures/squares/shapes/twoEqualTop.tikz}}}
		\hspace{2mm}
		\fbox{\resizebox{!}{3mm}{\input{pictures/squares/shapes/twoEqualBot.tikz}}}
		\hspace{2mm}
		\fbox{\resizebox{!}{3mm}{\input{pictures/squares/shapes/twoEqualBLTR.tikz}}}
		\hspace{2mm}
		\fbox{\resizebox{!}{3mm}{\input{pictures/squares/shapes/twoEqualTLBR.tikz}}}
	}_{\text{6 shapes, }{k \choose{3}} \cdot 3! \text{ each}}
	\hspace{2mm}
	\underbrace{
		\fbox{\resizebox{!}{3mm}{\begin{tikzpicture}

	\begin{pgfonlayer}{nodelayer}
		\node[dot,fill=black,draw=black] (tl) at (-0.75,+0.75) {};
		\node[dot,fill=black,draw=black] (bl) at (-0.75,-0.75) {};
		\node[dot,fill=black,draw=black] (tr) at (+0.75,+0.75) {};
		\node[dot,fill=black,draw=black] (br) at (+0.75,-0.75) {};
	\end{pgfonlayer}
	\begin{pgfonlayer}{edgelayer}
	\end{pgfonlayer}

\end{tikzpicture}}}
	}_{{k \choose{4}} \cdot 4!}
\end{equation}
By adding up the contributions from pieces of each shape, we get the following closed expression for the outcome probability  $\mathbb{P}[+|U]$:
\begin{equation}
	 \mathbb{P}[+|U] = \frac{1}{d^4}(\# U)^4
\end{equation}
For $d\geq 3$ we observe third-order interference, witnessed (by definition) by the following inequality:
\begin{eqnarray}
	\mathbb{P}[+|\{1,2,3\}] \neq
		\sum_{\stackrel{V \subset \{1,2,3\}}{\textnormal{s.t. }\#V = 2}}\hspace{-2mm}\mathbb{P}[+|V] \hspace{2mm}- \sum_{\stackrel{V \subset \{1,2,3\}}{\textnormal{s.t. }\#V = 1}}\hspace{-2mm}\mathbb{P}[+|V]
\end{eqnarray}
Indeed, the left hand side evaluates to $81/d^4$, while the right hand side evaluates to the following expression (again by standard combinatorial arguments):
\begin{equation}
	\frac{1}{d^4}\Big[ {3\choose{2}}2^4 - {3\choose{1}}1^4 \Big] = \frac{1}{d^4}45 \neq \frac{1}{d^4}81
\end{equation}
The difference between left and right hand sides is $36/d^4$, which is exactly the contribution $\frac{1}{d^4}6\cdot{3\choose{3}}\cdot3!$ of the 6 shapes requiring 3 distinct values (appearing in $P_{\{1,2,3\}}$ but not in any of the sub-projectors).
For $d\geq 4$ we observe fourth-order interference, witnessed (by definition) by the following inequality:
\begin{eqnarray}
	\mathbb{P}[+|\{1,2,3,4\}] \neq
		\sum_{\stackrel{V \subset \{1,2,3,4\}}{\textnormal{s.t. }\#V = 3}}\hspace{-2mm}\mathbb{P}[+|V]\hspace{2mm} - \sum_{\stackrel{V \subset \{1,2,3,4\}}{\textnormal{s.t. }\#V = 2}}\hspace{-2mm}\mathbb{P}[+|V] \hspace{2mm}+ \sum_{\stackrel{V \subset \{1,2,3,4\}}{\textnormal{s.t. }\#V = 1}}\hspace{-2mm}\mathbb{P}[+|V]
\end{eqnarray}
Indeed, the left hand side evaluates to $256/d^4$, while the right hand side evaluates to the following expression (again by standard combinatorial arguments):
\begin{equation}
	\frac{1}{d^4}\Big[ {4\choose{3}}3^4 - {4\choose{2}}2^4 + {4\choose{1}}1^4 \Big] = \frac{1}{d^4}232 \neq \frac{1}{d^4}256
\end{equation}
The difference between left and right hand sides is $24/d^4$, which is exactly the contribution $\frac{1}{d^4}{4\choose{4}}\cdot4!$ of the shape requiring 4 distinct values (appearing in $P_{\{1,2,3,4\}}$ but not in any of the sub-projectors).

For $d \geq 5$, however, we observe absence of fifth-order (or higher-order) interference, witnessed (by definition) by the following equality:
\begin{eqnarray}
	\mathbb{P}[+|\{1,2,3,4,5\}] &=
		\sum\limits_{\stackrel{V \subset \{1,2,3,4,5\}}{\textnormal{s.t. }\#V = 4}}\hspace{-2mm}\mathbb{P}[+|V] \hspace{2mm}- \sum\limits_{\stackrel{V \subset \{1,2,3,4,5\}}{\textnormal{s.t. }\#V = 3}}\hspace{-2mm}\mathbb{P}[+|V] \nonumber\\
		&+ \sum\limits_{\stackrel{V \subset \{1,2,3,4,5\}}{\textnormal{s.t. }\#V = 2}}\hspace{-2mm}\mathbb{P}[+|V] \hspace{2mm}- \sum\limits_{\stackrel{V \subset \{1,2,3,4,5\}}{\textnormal{s.t. }\#V = 1}}\hspace{-2mm}\mathbb{P}[+|V]
\end{eqnarray}
Indeed, the left hand side evaluates to $625/d^4$, and the right hand side yields the same:
\begin{equation}
	\frac{1}{d^4}\Big[ {5\choose{4}}4^4 - {5\choose{3}}3^4 + {5\choose{2}}2^4 - {5\choose{1}}1^4 \Big] = \frac{1}{d^4}625
\end{equation}


\section{Conclusions}
\label{section_conclusions}

In this work, we used an iterated CPM construction known as double-dilation to construct a full-fledged probabilistic theory of density hypercubes, possessing hyper-decoherence maps and showing higher-order interference effects. We have defined all the necessary categorical structures. We have gone over the mathematical detail of the (hyper-)decoherence–induced relationship between our new theory, quantum theory and classical theory. We have imported diagrammatic reasoning from the familiar setting of mixed-state quantum theory. We have developed a graphical formalism to study the internal component symmetries of states and processes. Finally, we have shown that the theory displays interference effects of orders up to four, but not of orders five and above. 

A number of questions are left open and will be answered as part of future work. Firstly, we endeavour to carry out a more physically-oriented analysis of the theory, including a study of the structure of normalised states and effects and a characterisation of the normalised reversible transformations. Secondly, we need to investigate the physical significance and implications of sub-normalisation of the hyper-decoherence maps, and construct a suitable extension of our theory where said maps become normalised. Finally, we intend to look at concrete implementations of certain protocols in our theory, such as those previously studied \cite{Lee-Selby-Grover,Lee-Selby-interference} in the context of higher-order interference. 

From a categorical standpoint, we also wish to further understand the specific roles played by double-mixing and double-dilation in our theory. At present, we know that the former is enough for density hypercubes to show higher-order interference and decohere to classical systems, but the latter seems to be necessary for quantum systems to arise by hyper-decoherence. Further investigation will hopefully shed more light on the individual contributions of the two constructions. Finally, we endeavour to investigate the generalisation of our results to higher iterated dilation, and more generally to higher-order CPM constructions \cite{higherOrderCPM} (with finite abelian symmetry groups other than the $\integersMod{2}^N$ groups arising from iterated dilation).

\newpage
\bibliographystyle{eptcs}
\bibliography{bibliography}

\appendix

\section{Proofs}
\label{appendix_proofs}

\setcounter{theorem_c}{\value{proposition_causality_c}}
\begin{proposition}
\label{proposition_causality}
	The process theory $\DDCategory{\fHilbCategory}$ is causal, in the following sense: for every object $\DDCategory{H}$, the only effect $\DDCategory{H} \rightarrow \reals^+$ in \DDCategory{\fHilbCategory} which yields the scalar $1$ on all normalised states of $\DDCategory{H}$ is the ``forest'' discarding map of density hypercubes $\trace{\,\,\DDCategory{H}}$. 
\end{proposition}
\begin{proof}
	Seen as an effect in \CPMCategory{\fHilbCategory}, any such effect must take the form of a sum $\sum_{x \in X} p_x \ket{a_x}\bra{a_x}$, where $p_x \in \reals^+$ and $(\ket{a_x})_{x \in X}$ is an orthonormal basis for ${H} \otimes {H}$ which satisfies an additional condition due to the symmetry requirement for effects in \DDCategory{\fHilbCategory}. If we write $\sigma_{{H},{H}}$ for the symmetry isomorphism ${H} \otimes {H} \rightarrow {H} \otimes {H}$ which swaps two copies of ${H}$ in $\fHilbCategory$, the additional condition on the orthonormal basis implies that for each $x \in X$ there is a unique $y \in X$ such that $\sigma_{{H},{H}} \ket{a_x} = e^{i\theta_{x}}\ket{a_y}$ and $p_x = p_y$; we define an involutive bijection $s: X \rightarrow X$ by setting $s(x)$ to be that unique $y$. For each $x \in X$, consider the normalised state $\rho_x := \frac{1}{2}(\ket{a_x}\bra{a_x} + \ket{a_{s(x)}}\bra{a_{s(x)}})$ in \CPMCategory{\fHilbCategory}, which we can realise in the sub-category \DDCategory{\fHilbCategory} by considering the classical structure $\hbox{\input{symbols/ZdotSym.tex}}\!$ on $\complexs^2$ corresponding to orthonormal basis $\ket{0}, \ket{1}$ and the vector $\ket{r_x} := \frac{1}{\sqrt[4]{2}}(\ket{a_x} \otimes \ket{0} + \ket{a_{s(x)}} \otimes \ket{1})$:
	\begin{equation}
		 \rho_x 
		\hspace{2mm} = \hspace{2mm}
		\frac{1}{2}
		\Bigg(
		\hspace{2mm}
 		\scalebox{0.8}{$
			\input{pictures/causalityProofEqn1l.tikz}
		$}
		\hspace{2mm}
		\Bigg)
		\hspace{3mm} = \hspace{3mm}
 		\scalebox{0.8}{$
			\input{pictures/causalityProofEqn1r.tikz}
		$}
	\end{equation}
	Now observe that the requirement that our effect yield $1$ on all normalised states implies, in particular, that the following equation must hold:
	\begin{equation}
		1
		\hspace{3mm} = \hspace{3mm}
		\scalebox{0.6}{$
			\input{pictures/causalityProofEqn2.tikz}
		$}
		\hspace{3mm} = \hspace{3mm}
		\frac{1}{2}(p_x+p_{s(x)})
		\hspace{3mm} = \hspace{3mm}
		p_x
	\end{equation}
	As a consequence, our effect is written $\sum_{x \in X} \ket{a_x}\bra{a_x}$, which is exactly the ``forest'' discarding map $\trace{\,\,\DDCategory{H}}$ of density hypercubes on $\DDCategory{H}$.
\end{proof}

\setcounter{theorem_c}{\value{proposition_classical_c}}
\begin{proposition}
\label{proposition_classical}
	Let $\KaroubiEnvelope{\DDCategory{\fHilbCategory}}$ be the Karoubi envelope of \DDCategory{\fHilbCategory}, and write $\KaroubiEnvelope{\DDCategory{\fHilbCategory}}_K$ for the full subcategory of $\KaroubiEnvelope{\DDCategory{\fHilbCategory}}$ spanned by objects in the form $(\DDCategory{H},\decoh{\hbox{\input{symbols/ZdotSym.tex}}\!})$. There is an $\reals^+$-linear monoidal equivalence of categories between $\KaroubiEnvelope{\DDCategory{\fHilbCategory}}_K$ and the probabilistic theory $\RMatCategory{\reals^+}$ of classical systems. Furthermore, classical stochastic maps correspond to the maps in $\KaroubiEnvelope{\DDCategory{\fHilbCategory}}_K$ normalised with respect to the discarding maps $\trace{\,\,(\DDCategory{H},\decoh{\hbox{\input{symbols/ZdotSym.tex}}\!})} := \trace{\,\,\DDCategory{H}} \circ \decoh{\hbox{\input{symbols/ZdotSym.tex}}\!}$, which we can write explicitly as follows:
	\newcounter{proposition_classical_c_eq_aux}
	\setcounter{proposition_classical_c_eq_aux}{\value{equation}}
	\setcounter{equation}{\value{proposition_classical_c_eq}}
	\begin{equation}
		\trace{\,\,(\DDCategory{H},\decoh{\hbox{\input{symbols/ZdotSym.tex}}\!})}
		\hspace{3mm} := \hspace{3mm}
		\scalebox{0.7}{$
			\input{pictures/DHclassicalDiscardingMaps1.tikz}
		$}
		\hspace{3mm} = \hspace{3mm}
		\scalebox{0.7}{$
			\input{pictures/DHclassicalDiscardingMaps2.tikz}
		$}
	\end{equation}
	\setcounter{equation}{\value{proposition_classical_c_eq_aux}}
\end{proposition} 
\begin{proof}
Consider two objects $(\DDCategory{H},\decoh{\hbox{\input{symbols/DdotSym.tex}}\!\!})$ and $(\DDCategory{K},\decoh{\hbox{\input{symbols/YbwdotSym.tex}}\!\!})$, where $\hbox{\input{symbols/DdotSym.tex}}\!\!$ and $\hbox{\input{symbols/YbwdotSym.tex}}\!\!$ are special commutative $\dagger$-Frobenius algebras associated with orthonormal bases $(\ket{\psi_x})_{x \in X}$ and $(\ket{\phi_y})_{y \in Y}$ of ${H}$ and ${K}$ respectively. The morphisms $(\DDCategory{H},\decoh{\hbox{\input{symbols/DdotSym.tex}}\!\!}) \rightarrow (\DDCategory{K},\decoh{\hbox{\input{symbols/YbwdotSym.tex}}\!\!})$ in $\KaroubiEnvelope{\DDCategory{\fHilbCategory}}$ are exactly the maps of density hypercubes $\DDCategory{H} \rightarrow \DDCategory{K}$ in the following form:
\begin{equation}
	\scalebox{0.6}{$
		\input{pictures/DHclassicalMaps.tikz}
	$}
\end{equation}
We can expand the definition of decoherence maps to see that these morphisms correspond to generic matrices $M_{xy}$ of non-negative real numbers, with matrix composition as sequential composition, Kronecker product as tensor product, and the $\reals^+$-linear structure of matrix addition. 
\begin{equation}
	\scalebox{0.6}{$
		\input{pictures/DHclassicalMaps.tikz}
	$}
	\hspace{3mm} = \hspace{3mm}
	\sum_{x \in X} \sum_{y \in Y} \hspace{3mm}
	\scalebox{0.6}{$
		\input{pictures/DHclassicalMapsExplicit.tikz}
	$}
\end{equation}
The discarding maps obtained by decoherence of the environment structure for \DDCategory{\fHilbCategory} yield the usual environment structure for classical systems:
\begin{equation}
	\trace{\,(\mathcal{H},\decoh{\hbox{\input{symbols/ZdotSym.tex}}\!})}
	\hspace{3mm} := \hspace{3mm}
	\scalebox{0.7}{$
		\input{pictures/DHclassicalDiscardingMaps1.tikz}
	$}
	\hspace{3mm} = \hspace{3mm}
	\scalebox{0.7}{$
		\input{pictures/DHclassicalDiscardingMaps2.tikz}
	$}
	\hspace{3mm} = \hspace{3mm}
	\sum_{x \in X} \hspace{3mm}
	\scalebox{0.7}{$
		\input{pictures/DHclassicalDiscardingMaps3.tikz}
	$}
\end{equation}
Hence $\mathcal{C}_K$ is equivalent to the probabilistic theory $\RMatCategory{\reals^+}$ of classical systems. 
\end{proof}

\setcounter{theorem_c}{\value{proposition_quantum_c}}
\begin{proposition}
\label{proposition_quantum}
	Let $\KaroubiEnvelope{\DDCategory{\fHilbCategory}}$ be the Karoubi envelope of \DDCategory{\fHilbCategory}, and write $\KaroubiEnvelope{\DDCategory{\fHilbCategory}}_Q$ for the full subcategory of $\KaroubiEnvelope{\DDCategory{\fHilbCategory}}$ spanned by objects in the form $(\DDCategory{H},\hypdecoh{\hbox{\input{symbols/ZdotSym.tex}}\!})$. There is an $\reals^+$-linear monoidal equivalence of categories between $\KaroubiEnvelope{\DDCategory{\fHilbCategory}}_Q$ and the probabilistic theory $\CPMCategory{\fHilbCategory}$ of quantum systems and CP maps between them. Furthermore, trace-preserving CP maps correspond to the maps in $\KaroubiEnvelope{\DDCategory{\fHilbCategory}}_Q$ normalised with respect to the discarding maps $\trace{\,\,(\DDCategory{H},\hypdecoh{\hbox{\input{symbols/ZdotSym.tex}}\!})} := \trace{\,\,\DDCategory{H}} \circ \hypdecoh{\hbox{\input{symbols/ZdotSym.tex}}\!}$, which we can write explicitly as follows:
	\newcounter{proposition_quantum_c_eq_aux}
	\setcounter{proposition_quantum_c_eq_aux}{\value{equation}}
	\setcounter{equation}{\value{proposition_quantum_c_eq}}
	\begin{equation}
		\trace{\,\,(\DDCategory{H},\hypdecoh{\hbox{\input{symbols/ZdotSym.tex}}\!})}
		\hspace{3mm} := \hspace{3mm}
		\scalebox{0.7}{$
			\input{pictures/DHquantumDiscardingMaps1.tikz}
		$}
		\hspace{3mm} = \hspace{3mm}
		\scalebox{0.7}{$
			\input{pictures/DHclassicalDiscardingMaps2.tikz}
		$}
	\end{equation}
	\setcounter{equation}{\value{proposition_quantum_c_eq_aux}}
\end{proposition} 
\begin{proof}
We can define an essentially surjective, faithful monoidal functor from \KaroubiEnvelope{\DDCategory{\fHilbCategory}} to the category \CPMCategory{\fHilbCategory} of quantum systems and CP maps by setting $(\DDCategory{H},\hypdecoh{\hbox{\input{symbols/DdotSym.tex}}\!\!}) \mapsto \mathcal{H}$ on objects and doing the following on morphisms:
\begin{equation}
	\scalebox{0.6}{$
		\input{pictures/DHQuantumSysSubcat1.tikz}
	$}
	\hspace{3mm} \mapsto \hspace{3mm}
	\scalebox{0.6}{$
		\input{pictures/DHQuantumSysSubcat2.tikz}
	$}
\end{equation}
In order to show monoidal equivalence we need to show that the functor is also full, i.e. that every CP map can be obtained from a map of \KaroubiEnvelope{\DDCategory{\fHilbCategory}} in this way. Because of compact closure, it is actually enough to show that all states can be obtained this way. Consider a finite-dimensional Hilbert space $H$ and a classical structure $\hbox{\input{symbols/YbwdotSym.tex}}\!\!$ on it, and write $(\ket{\psi_x})_{x \in X}$ for the orthonormal basis of $H$ associated to $\hbox{\input{symbols/YbwdotSym.tex}}\!\!$. The most generic mixed quantum state on $\mathcal{H}$ takes the form $\rho = \sum_{y \in Y} p_y \ket{\gamma_y}\bra{\gamma_y}$, where $(\ket{\gamma_y})_{y \in Y}$ is some orthonormal basis of ${H}$ and $p_y \in \reals^+$. Let $\hbox{\input{symbols/ZdotSym.tex}}\!$ be the classical structure associated with the orthonormal basis $(\ket{\gamma_y})_{y \in Y}$, and define the states $\ket{\sqrt[\hbox{\input{symbols/YbwdotSym.tex}}\!\!]{\gamma_y}}:= \sum_{x \in X} \ket{\psi_x} \sqrt{\braket{\psi_x}{\gamma_y}}$, where $\sqrt{\braket{\psi_x}{\gamma_y}} \in \complexs$ is such that $\sqrt{\braket{\psi_x}{\gamma_y}}^2 = \braket{\psi_x}{\gamma_y} \in \complexs$. If we write $\ket{\phi} := \sum_{yY} \sqrt{p_y} \ket{\sqrt[\hbox{\input{symbols/YbwdotSym.tex}}\!\!]{\gamma_y}} \otimes \ket{\gamma_y}$, then the desired state $\rho$ can be obtained as follows:
\begin{equation}
	\rho
	\hspace{1mm} = \hspace{1mm}
	\sum_{y \in Y} \hspace{1mm} p_y \hspace{2mm}
	\scalebox{0.8}{$
		\begin{tikzpicture}
	\begin{pgfonlayer}{nodelayer}
		\node[style=state] (psi) at (0,0) {$\Gamma_y$};
		\node (out) at (2,0) {};
	\end{pgfonlayer}
	\begin{pgfonlayer}{edgelayer}
		\draw[-] (psi) to (out);	
	\end{pgfonlayer}
\end{tikzpicture}
	$}
	\hspace{1mm} =\hspace{3mm}
	\sum_{y \in Y} 
	\hspace{1mm} 
	\raisebox{4mm}{$\sqrt{p_{y}}$}
	\hspace{-7mm}
	\raisebox{-4mm}{$\sqrt{p_{y}}$}
	\hspace{2mm}
	\scalebox{0.8}{$
		\input{pictures/fullOnStates2.tikz}
	$}
	\hspace{2mm} =\hspace{3mm}
	\scalebox{0.8}{$
		\input{pictures/fullOnStates3.tikz}
	$}
\end{equation}	
Hence the monoidal functor defined above is full, faithful and essentially surjective, i.e. an equivalence of categories. Furthermore, it is $\reals^+$-linear and it respects discarding maps.
\end{proof}

\section{Possibility of extension for the theory of density hypercubes}
\label{appendix_extension}

The theory of density hypercubes presented in this work is fully-fledged\footnote{In the sense that it contains all the features necessary to consistently talk about operational scenarios, such as preparations, measurements, controlled transformations, reversible transformation, test, non-locality scenarios, etc.} but incomplete: as shown by Equation \ref{proposition_quantum_eq_label}, the hyper-decoherence maps are not normalised (i.e. they are not ``deterministic'', in the parlance of OPTs/GPTs)
	\setcounter{proposition_quantum_c_eq_aux}{\value{equation}}
	\setcounter{equation}{\value{proposition_quantum_c_eq}}
	\begin{equation}
		\trace{\,\,(\DDCategory{H},\hypdecoh{\hbox{\input{symbols/ZdotSym.tex}}\!})}
		\hspace{3mm} := \hspace{3mm}
		\scalebox{0.7}{$
			\input{pictures/DHquantumDiscardingMaps1.tikz}
		$}
		\hspace{3mm} = \hspace{3mm}
		\scalebox{0.7}{$
			\input{pictures/DHclassicalDiscardingMaps2.tikz}
		$}
	\end{equation}
	\setcounter{equation}{\value{proposition_quantum_c_eq_aux}}
When it comes to this work, however, this is not much of a problem: all we need to show is that an extension of our theory can exists in which the ``tree-on-a-bridge'' effect above can be completed to the discarding map, and our results---both hyper-decoherence to quantum theory and higher-order interference---will automatically apply to any such extension.

Let $(\ket{\psi_x})_{x \in X}$ be the orthonormal basis associated with the special commutative $\dagger$-Frobenius algebra $\hbox{\input{symbols/ZbwdotSym.tex}}\!\!$. The effect needed to complete $\trace{\,\,(\DDCategory{H},\hypdecoh{\hbox{\input{symbols/ZdotSym.tex}}\!})}$ to the discarding map $\trace{\,\,\DDCategory{H}}$ is itself an effect in $\CPMCategory{\fdHilbCategory}$, which can be written explicitly as follows:
\begin{equation}
	\label{eq_extensioneffect}
	\scalebox{0.7}{$\input{pictures/DHforestdiscardingmap.tikz}$}
	\hspace{3mm} - \hspace{3mm}
	\scalebox{0.7}{$
		\input{pictures/DHclassicalDiscardingMaps2.tikz}
	$}
	\hspace{3mm} = \hspace{3mm}
	\sum_{\substack{x,y \in X \\ \text{ s.t. } x \neq y}} \hspace{3mm}
	\scalebox{0.7}{$
		\input{pictures/DHextensioneffect.tikz}
	$}
\end{equation}
Because it is an effect in $\CPMCategory{\fdHilbCategory}$, which has $\reals^+$ as its semiring of scalars, it is in particular non-negative on all states in $\DDCategory{\fdHilbCategory}$, showing that: (i) hyper-decoherence maps are sub-normalised; (ii) our theory does not satisfy the no-restriction condition; (iii) an extension to a theory with normalised hyper-decoherence is possible. This shows that our results on hyper-decoherence have physical significance.
Furthermore, let $\ket{1},...,\ket{d}$ be an orthonormal basis of $\complexs^d$, and let $\hbox{\input{symbols/ZbwdotSym.tex}}\!\!$ correspond to the Fourier basis for the finite abelian group $\integersMod{d}$:
\begin{equation}
	\bigg(\frac{1}{\sqrt{d}}\sum_{j=1}^d e^{i\frac{2\pi}{d}jk} \ket{j}\bigg)_{k=1,...,d}
\end{equation} 
Choosing $k:=d$, in particular, shows that the orthonormal basis above contains the state $\frac{1}{\sqrt{d}}\ket{\psi_+}$ used in Section \ref{section_higherOrderInterference}. Then the effect defined in Equation \ref{eq_extensioneffect} also shows that the computation of $\mathbb{P}[+|U]$ in Section \ref{section_higherOrderInterference} can be done as part of a bona-fide measurement in any such extended theory, and hence that our higher-order interference result has physical significance.

\end{document}